\documentclass[12pt, draftclsnofoot, onecolumn]{IEEEtran}

\usepackage{amsthm}
\usepackage[cmex10]{amsmath}
\usepackage{graphicx}
\usepackage{color}
\usepackage[tight,footnotesize]{subfigure}
\usepackage{amsfonts}
\usepackage{algorithmic}
\usepackage{algorithm}
\usepackage{graphicx}
\usepackage{subfigure}
\usepackage{bm}
\usepackage{times,color,amssymb,epsfig,psfrag,stfloats,enumerate}
\usepackage{cite}
\usepackage{flushend}

\newtheorem{theorem}{Theorem}
\newtheorem{remark}{Remark}
\newtheorem{lemma}{Lemma}

\newtheorem{proposition}{Proposition}

\ifCLASSINFOpdf
\else
\fi
\begin{document}
%
\title{Flight Time Minimization of UAV for Data Collection over Wireless Sensor Networks}

\author{Jie~Gong,~\IEEEmembership{Member,~IEEE}, Tsung-Hui Chang,~\IEEEmembership{Senior Member,~IEEE}, Chao Shen,~\IEEEmembership{Member,~IEEE}, Xiang~Chen,~\IEEEmembership{Member,~IEEE}

\thanks{J. Gong is with School of Data and Computer Science, Sun Yat-sen University, Guangzhou 510006, China. Email: gongj26@mail.sysu.edu.cn.}
\thanks{T.-H. Chang is with School of Science and Engineering, The Chinese University of Hong Kong, Shenzhen, Shenzhen 518172, China. Email: tsunghui.chang@ieee.org}
\thanks{C. Shen is with State Key Lab of Rail Traffic Control and Safety, Beijing Jiaotong University, Beijing, China. Email: chaoshen@bjtu.edu.cn}
\thanks{X. Chen is with the School of Electronics and Information Engineering,  Sun Yat-sen University, Guangzhou 510006, China, SYSU-CMU Shunde International Joint Research Institute and Key Lab of EDA, Research Institute of Tsinghua University in Shenzhen, Shenzhen 518057, China. Email: chenxiang@mail.sysu.edu.cn.}
}

\maketitle

\begin{abstract}
In this paper, we consider a scenario where an unmanned aerial vehicle (UAV) collects data from a set of sensors on a straight line. The UAV can either cruise or hover while communicating with the sensors. The objective is to minimize the UAV's total flight time from a starting point to a destination while allowing each sensor to successfully upload a certain amount of data using a given amount of energy. The whole trajectory is divided into non-overlapping data collection intervals, in each of which one sensor is served by the UAV. The data collection intervals, the UAV's speed and the sensors' transmit powers are jointly optimized. The formulated flight time minimization problem is difficult to solve. We first show that when only one sensor is present, the sensor's transmit power follows a water-filling policy and the UAV's speed can be found efficiently by bisection search. Then, we show that for the general case with multiple sensors, the flight time minimization problem can be equivalently reformulated as a dynamic programming (DP) problem. The subproblem involved in each stage of the DP reduces to handle the case with only one sensor node. Numerical results present insightful behaviors of the UAV and the sensors. Specifically, it is observed that the UAV's optimal speed is proportional to the given energy of the sensors and the inter-sensor distance, but inversely proportional to the data upload requirement.
\end{abstract}

\IEEEpeerreviewmaketitle

\renewcommand{\algorithmicrequire}{\textbf{Input:}}
\renewcommand{\algorithmicensure}{\textbf{Output:}}

\section{Introduction}
Recently, wireless communication with unmanned aerial vehicles (UAVs) \cite{zeng2016wireless} has been considered as a promising technology to expand network coverage and enhance system throughput, by leveraging the UAVs' high mobility \cite{Dji2017inspire} and line-of-sight (LOS) dominated air-ground channels \cite{sun2015dual}. One of the key applications is wide-area data collection in wireless sensor networks \cite{Dong2014uav}. Conventionally, each sensor node delivers their monitored data to a fusion center via multi-hop transmissions. Hence, a sensor node requires to not only transmit its own data, but also relay the others'. As a consequence, the sensors' battery may drain quickly and the multi-hop network connection may be lost. By using the UAVs as mobile fusion centers, every sensor node can directly send observed data to a UAV. In addition, the LOS channel condition results in higher data rate for ground-to-air transmissions compared with ground-to-ground transmissions. However, as UAVs are energy constrained due to the limited on-board battery, it is paramount to shorten the flight time needed for a data collection mission.

Different from the conventional communication techniques, there is a trajectory optimization issue for UAV-aided wireless communications. To improve network connectivity, UAVs' deployment and movement were optimized to track the network topology in \cite{han2009opt}. In \cite{Grotli2012path}, offline path planning of UAVs was addressed for collision avoidance and fuel efficiency. Joint UAV deployment and trajectory optimization problem was solved in \cite{erdem2017deploy} with a quantization theory approach, and joint trajectory and communication power control for multiple UAVs was studied in \cite{we2017joint}. In addition, UAVs are widely used as mobile relays. Reference \cite{zeng2016throughput} studied the throughput maximization problem for a UAV relay and showed that the uplink power of users should follow a ``staircase" water filling structure. In \cite{kalogerias2016mobile}, joint optimization of multi-UAV beamforming and relay positions for throughput maximization was studied based on stochastic optimization techniques. A round trip ``load-carry-and-deliver" protocol was tested and evaluated by experiments in \cite{cheng2007maximizing}. Besides serving as relays, UAVs can also be used as mobile base stations (BSs) for emergent communications. In \cite{lyu2017placement} and \cite{alzenad20173d}, BS placement was optimized in the two-dimensional (2D) space and three-dimensional (3D) space, respectively, to minimize the required number of mobile BSs while maximizing their coverage. The coverage of UAVs as mobile BSs was analytically studied in \cite{moza2016efficient} considering inter-UAV interference and beamwidth design.

In addition, there has been a growing research interest in applying UAV for data collection and dissemination in wireless sensor networks. The aerial link characterization based on practical protocols and experiments was given in \cite{ahmed2016on}. Reference \cite{jiang2012optimization} considered data collection via uplink transmission, and proposed to mitigate multi-sensor interference by adjusting the UAV heading and beamforming. Adaptive modulation strategy was adopted in \cite{Adbulla2014optimal} to improve energy efficiency of sensor nodes while guaranteeing user fairness. To avoid contention due to simultaneous data transmissions from multiple sensor nodes, a priority-based frame selection scheme was proposed in \cite{say2016priority}. In \cite{zhan2017energy}, wake-up and sleep adaptation was applied for sensor nodes and UAV's trajectory optimization was jointly considered to minimize the sensors' energy consumption. One dimensional information dissemination problem is considered in \cite{lyu2016cyclical}, where the sensors are served in cyclical TDMA mode, and the service regions for all sensors are optimized.

It is worthwhile to note that most of the existing works mentioned above focus on enhancing energy efficiency or spectrum efficiency of sensor nodes, but overlook the fact that the limited energy of UAVs is one of the fundamental bottlenecks in UAV-aided wireless networks. As a matter of fact, the dominant energy consumption of a UAV lies in the propulsion control system that accelerates the UAV and maintains its flight height. In \cite{franco2015energy}, a UAV's energy consumption was modeled as a function of flight speed and operation conditions such as climbing, hovering, and so on. A UAV trajectory optimization problem with detailed propulsion energy consumption considering both velocity and acceleration was studied in \cite{zeng2017energy}. However, as the UAV's energy consumption model is quite complex, the problems are difficult to be optimally solved. {Intuitively, under a certain constraint on the flight speed, flight time minimization is an alternative for energy consumption minimization \cite{franco2015energy}. It arises in practical applications such as mission completion time minimization \cite{zeng2018traj}.} { Two dual problems: data delivering maximization under a maximum flight time constraint, and flight time minimization under given load requirement were solved in \cite{moza2017wireless} based on optimal transport theory. In \cite{moza2018drone}, drone-based antenna array approach with multiple UAVs was proposed to significantly reduce the service time for ground users.} In cellular networks, the same problem subject to a link quality constraint between ground base station (GBS) and UAV was studied in \cite{zhang2017cellular}. However, in wireless sensor networks, there is still a lack of research efforts to consider both the UAV and sensors' energy consumption as well as the quality of service of the sensor nodes at the same time.


In this paper, we study a flight time minimization problem for a UAV which collects data from a set of energy constrained ground sensors. Each of the sensors wants to upload a certain amount of data to the UAV. The UAV can collect data either during flying or hovering. We assume that the sensors are located on a line and the UAV's trajectory is divided into non-overlapping data collection intervals, each of which is dedicated to data collection from one sensor node. The objective is to minimize the total flight time of the UAV from an initial point to a destination by jointly optimizing the division of intervals, the UAV's speed, as well as the sensors' transmission power. The contributions of this paper are summarized as follows.

\begin{itemize}
\item The formulated UAV flight time minimization problem is intrinsically difficult. We first consider the single-sensor scenario. While the problem is still difficult when only one sensor node is present, we reveal some insightful structures for the optimal solution. Specifically, we present an explicit condition on the feasibility of the problem. When the problem is feasible and if the data collection interval is given, we show that the optimal power allocation of the sensor follows a water-filling solution and the optimal speed can be efficiently obtained via bisection search. The data collection interval can be numerically found via a two-dimensional search.
\item The algorithm for solving the single-sensor case can be extended for solving the general scenario with multiple sensor nodes. In particular, by judiciously exploiting the problem structure, we show that the flight time minimization problem with multiple sensors can be equivalently formulated as a dynamic programming (DP) problem. In each stage of the DP, the optimal data collection interval for one sensor node is searched and the algorithm for the single-sensor case is used for finding the UAV's optimal speed and sensor's transmission power.
\item Numerical results illustrate the optimal behaviors of the UAV and the sensor nodes under different scenarios. In particular, the UAV's optimal speed is proportional to the sensors' energy budgets and the inter-sensor distance, but inversely proportional to the amount of data to upload. For the randomly distributed sensors with random amount of data and energy, the average minimum flight time increases with the average amount of data and decreases with the average amount of available energy.
\end{itemize}

The rest of the paper is organized as follows. Section \ref{sec:model} presents the system model and the problem formulation. Section \ref{sec:single} studies the single-sensor case. Then, the multi-sensor problem is solved in Section \ref{sec:multi}. Simulations are shown in Section \ref{sec:sim}. Finally, Section \ref{sec:concl} concludes the paper.

\section{System Model and Problem Formulation} \label{sec:model}
As shown in Fig. 1, we consider a scenario where a UAV flies over a set of $N$ sensors for data collection. The sensors are located on a line, labeled by $S_1, S_2, \cdots, S_N$. { The line model is motivated by the application scenarios such as a sensor network deployed along a highway, railway or power line.} Each sensor $n$ needs to upload $B_n$ information bits and is subject to a total energy budget $E_n$, where $n = 1, 2, \cdots, N$. The UAV flies at a fixed height $H$ from an initial point $S_0$ to a destination $S_{N+1}$, and applies time division protocol to sequentially receive the uplink data from the sensors. Specifically, the whole range $[S_0, S_{N+1}]$ is divided into $N$ non-overlapping intervals $[x_n, y_n], n = 1, 2, \cdots, N$ satisfying $S_0 \le x_1 \le y_1 \le x_2 \le y_2 \le \cdots \le x_N \le y_N \le S_{N+1}$. Each sensor node $n$ uploads its data when the UAV flies in the interval $[x_n, y_n]$. If $x_n = y_n$, the UAV hovers above the location $x_n$ and receives the data from sensor $n$. Otherwise, we assume the UAV flies with a constant speed $0 < v_n \le v_{\max}$ from $x_n$ to $y_n$ and receives the data during  flying. As no sensor uploads data in the interval $(y_n, x_{n+1})$, the UAV flies with the maximum speed $v_{\max}$ in order to minimize the total flight time. In this paper, the UAV's acceleration/deceleration process is ignored for analytical tractability.

\begin{figure}[th]
\centering
\includegraphics[width=4.5in]{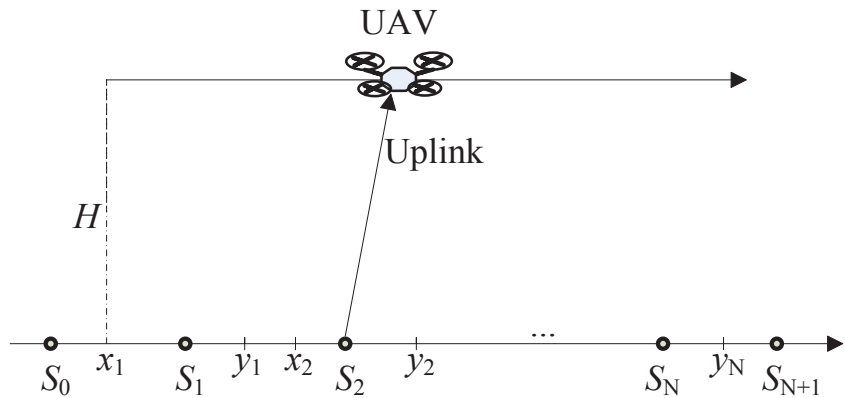}
\caption{Data collection by a UAV from ground sensors along a line.} \label{fig:systemM}
\end{figure}

\subsection{Data Collection Modes}
Since the UAV can receive data when either flying or hovering, we respectively consider the data collection models for the two cases.

\subsubsection{Data Collection during Flying} \label{sec:aviation}
If $v_n > 0$ and $x_n < y_n$, the UAV collects data from the ground node $S_n$ during $[x_n, y_n]$. The flight time or the data collection time is $t_n = (y_n-x_n)/v_n$. As the transmission distance changes during flying, the transmit power and the data rate should also adapt to the varying path-loss. The LOS ground-to-air channel model between the UAV and the sensors with pathloss exponent $\alpha \ge 2$ is adopted \cite{ahmed2016on}. With this model, the instantaneous data rate in the transmission interval $[x_n, y_n]$ is given by
\begin{align}
R_n(t) = \frac{1}{2}W\log_2 \bigg( 1+ \frac{p_n(t)\beta}{((x_n + v_nt - S_n)^2 + H^2)^{\frac{\alpha}{2}}}\bigg),
\end{align}
for $t \in [0, t_n]$ where $W$ is the bandwidth, $\beta$ is the reference signal-to-noise ratio (SNR) at the reference distance 1 meter, and $p_n(t)$ is the transmission power of the $n$th sensor, which satisfies the total energy constraint
\begin{align}
\int_{0}^{t_n} p_n(t) \mathrm d t \le E_n. \label{eq:Econstr}
\end{align}
{In this paper, we ignore the circuit power of the sensors and consider only the dominant transmission power. The derived results, however, can be easily extended to that with circuit power.} Besides, since each sensor $n$ requires to upload $B_n$ bits when the UAV flies over $[x_n, y_n]$, we have the data upload constraint as
\begin{align}
\int_{0}^{t_n} R_n(t) \mathrm d t \ge B_n. \label{eq:Bconstr}
\end{align}

Notice that there is a feasibility issue for data collection, i.e., with a given amount of sensor's energy $E_n$, is it feasible to upload $B_n$ bits within the time duration $t_n$? Since the best channel quality is experienced when the UAV is hovering right above the sensor $n$, the maximum number of data bits $B_n$ is related to the hovering mode, which is detailed below.

\subsubsection{Data Collection when Hovering} \label{sec:hover}
If $v_n = 0$ and $x_n = y_n$, the UAV hovers above location $x_n$. {Since the channel is unchanged, it is preferred for sensor $n$ to upload data with constant transmit power and data rate.} Denote the transmission time when hovering above the location $x_n$ by $t_n = T_{\mathrm h, n}(x_n)$. As the transmission link is static, $p_n(t)$ should be a constant in this case. Thus, sensor $n$'s energy constraint (\ref{eq:Econstr}) is simplified as
\begin{align}
T_{\mathrm h, n}(x_n)p_n(t) \le E_n.
\end{align}
To fully utilize sensor $n$'s energy budget to minimize the flight time, the transmission power should be maximized, i.e., $p_n(t) = E_n/T_{\mathrm h, n}(x_n)$. Then, the data constraint (\ref{eq:Bconstr}) is simplified as
\begin{align}
\frac{1}{2} T_{\mathrm h, n}(x_n) W \log_2 \bigg( 1+ \frac{\beta E_n}{T_{\mathrm h, n}(x_n) ((x_n-S_n)^2 + H^2)^{\frac{\alpha}{2}}}\bigg) \ge B_n. \label{eq:hover}
\end{align}
The function on the left hand side of (\ref{eq:hover}) has the following property.

\begin{lemma} \label{lemma:inc}
The function $f(x) = x \log_2 (1+ \frac{a}{x}), a > 0, x > 0$ is an increasing function, and $f(x) < \frac{a}{\ln2}$.
\end{lemma}
\begin{proof}
See Appendix \ref{proof:inc}.
\end{proof}

Based on Lemma \ref{lemma:inc}, the left hand side of (\ref{eq:hover}) is an increasing function of $T_{\mathrm h, n}(x_n)$. Hence, the minimum $T_{\mathrm h, n}(x_n)$ satisfies (\ref{eq:hover}) with equality, i.e.,
\begin{align}
\frac{1}{2} T_{\mathrm h, n}(x_n) W \log_2 \bigg( 1+ \frac{\beta E_n}{T_{\mathrm h, n}(x_n) ((x_n-S_n)^2 + H^2)^{\frac{\alpha}{2}}}\bigg) = B_n. \label{eq:hovereq}
\end{align}
The above transcendental equation can be effectively solved by either line search or bisection search. As the UAV experiences the best channel condition when hovering on top of the user ($x_n = S_n$), the feasibility condition can be derived based on (\ref{eq:hovereq}) as follows.

\begin{proposition} \label{prop:feas}
\textbf{(Feasibility)} For each sensor $n$, the data constraint (\ref{eq:Bconstr}) is feasible if and only if
\begin{align}
B_n < \frac{W\beta E_n}{2H^{\alpha} \ln 2}. \label{cond:feas}
\end{align}
\end{proposition}
\begin{proof}
See Appendix \ref{proof:feas}.
\end{proof}

Hovering mode may be needed when the amount of information bits is large or the amount of sensor's energy is small. {However, collecting data when hovering may not be the most time efficient strategy when comparing to collecting data during flying.} Therefore, flying and hovering modes have to be selected depending on the values of $B_n$ and $E_n$, which will be incorporated in our problem formulation.

\subsection{Problem Formulation}
In this paper, we aim to minimize the total flight time while guaranteeing that all the sensors' data are successfully collected. If (\ref{cond:feas}) holds for all $n = 1, 2, \cdots, N$, i.e., the data collection is feasible for all the sensors, the problem can be formulated as
\begin{subequations}\label{prob:multisensor}
\begin{align}
\min_{\bm x, \bm y, \bm v, \bm p(t)} \;&\; \frac{(S_{N+1}-S_0) - \sum_{n=1}^N (y_n-x_n)}{v_{\max}} + \sum_{n=1}^N t_n \label{eq:tmin}\\
\mathrm{s.t.} \quad &\; (\ref{eq:Econstr}) \textrm{~and~} (\ref{eq:Bconstr}), \;\forall n, \nonumber\\
\;&\; S_0 \le x_1 \le y_1 \le x_2 \le \cdots \le y_N \le S_{N+1}, \label{eq:series}\\
\;&\; t_n = \frac{y_n - x_n}{v_n}I_{x_n \neq y_n} + T_{\mathrm h, n}(x_n) I_{x_n = y_n}, \;\forall n, \label{eq:time}\\
\;&\; 0 \le v_n \le v_{\max}, \;\forall n, \\
\;&\; p_n(t) \ge 0, \;\forall n, t,
\end{align}
\end{subequations}
where the optimization variables are the locations $\bm x = \{x_1, x_2, \cdots, x_N\}, \bm y = \{y_1, y_2, \cdots, y_N\}$, the UAV's speeds $\bm v = \{v_1, v_2, \cdots, v_N\}$, and the transmission power $\bm p(t) = \{p_1(t), p_2(t), \cdots, p_N(t)\}$. The function $I_{\mathrm{event}}$ is an indicator which equals 1 if the event is true and equals 0 otherwise. It can be seen that the interval variables $\bm x$, $\bm y$ for the sensors are coupled in the constraint (\ref{eq:series}), which makes (\ref{prob:multisensor}) difficult to solve. To tackle the problem, we firstly consider a single-sensor case, and then show how the solution of the single-sensor case can be extended to the general multi-sensor case in (\ref{prob:multisensor}).

\begin{figure}
\centering
\includegraphics[width=4.5in]{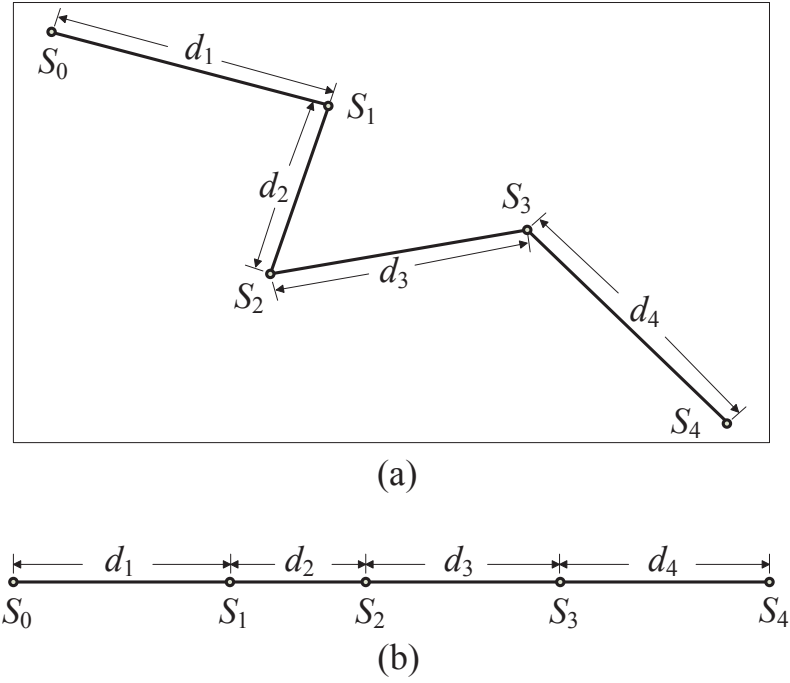}
\caption{Example of a 2D space problem represented by a line model: (a) a data collection problem in the 2D space with a given visit order; (b) equivalent representation by a line model.} \label{fig:2Dappl}
\end{figure}

{ \begin{remark}
We should emphasize here that the proposed solution for the flight time minimization problem under a line model can be applied to scenarios with sensors in the 2D space. In particular, the flight time minimization problem for the 2D space model is the same as that for a line model as long as the visit order of the sensors is given a priori, as illustrated in Fig.~\ref{fig:2Dappl}. Thus, the solution proposed for the line model can be applied to the 2D space problem by simply changing UAV's flight direction above each sensor.
\end{remark}

}

\section{Flight Time Minimization for Single-sensor Case}  \label{sec:single}
For the single-sensor case $N = 1$, without loss of generality, we set $S_1 = 0$ (origin point in the horizontal axis) and ignore the sensor index for all notations. Then, the problem (\ref{prob:multisensor}) reduce to
\begin{subequations}\label{prob:singlesensor}
\begin{align}
\min_{x, y, v, p(t)} \;&\; \frac{(S_2-S_0) - (y - x)}{v_{\max}} + t \label{eq:tmin_s}\\
\mathrm{s.t.} \quad &\; \int_{0}^{t} \frac{1}{2}W \log_2 \bigg( 1+ \frac{p(\tau)\beta}{((x + v\tau)^2 + H^2)^{\frac{\alpha}{2}}}\bigg) \mathrm d \tau \ge B, \label{eq:bits_s}\\
\;&\; \int_{0}^{t} p(\tau) \mathrm d \tau \le E, \label{eq:energy_s}\\
\;&\; S_0 \le x \le y \le S_2, \\
\;&\; t = \frac{y-x}{v} I_{x \neq y} + T_{\mathrm h}(x) I_{x = y}, \label{eq:time_s}\\
\;&\; 0 \le v \le v_{\max}, \label{eq:v_s}\\
\;&\; p(t) \ge 0. \label{eq:pt_s}
\end{align}
\end{subequations}

Since the hovering mode has been studied in the previous section, we mainly focus on the flying mode with $v > 0$ and $x < y$. The problem with only the flying mode for the single-sensor case is given by
\begin{subequations}\label{prob:single_nohover}
\begin{align}
\min_{x, y, v, p(t)} \;&\; \frac{S_2-S_0}{v_{\max}} + (y-x)\left(\frac{1}{v} - \frac{1}{v_{\max}} \right) \label{eq:tmin_sn}\\
\mathrm{s.t.} \quad &\; \int_{0}^{\frac{y-x}{v}} \frac{1}{2}W \log_2 \bigg( 1+ \frac{p(\tau)\beta}{((x + v\tau)^2 + H^2)^{\frac{\alpha}{2}}}\bigg) \mathrm d \tau \ge B, \label{eq:bits_sn}\\
\;&\; \int_{0}^{\frac{y-x}{v}} p(\tau) \mathrm d \tau \le E, \label{eq:energy_sn}\\
\;&\; S_0 \le x < y \le S_2, \\
\;&\; 0 < v \le v_{\max},  \; p(t) \ge 0. \nonumber
\end{align}
\end{subequations}

The problem (\ref{prob:single_nohover}) includes the power allocation optimization over $p(t)$, the UAV's speed optimization over $v$, and the data upload interval optimization over $x$ and $y$. These subproblems are solved separately as follows.

\subsection{Power Allocation} \label{sec:pa}
Suppose that the upload interval $[x,y]$ and the UAV's speed $v$ are fixed and given. It is obvious that to minimize the flight time, the sensor should allocate its power to maximize the uplink throughput on the left hand side of (\ref{eq:bits_sn}). Thus, let us consider the following throughput maximization problem
\begin{subequations}\label{prob:thrmax}
\begin{align}
\max_{p(\tau) \ge 0} \;&\; \int_0^{\frac{y-x}{v}} \frac{1}{2}W \log_2 \left( 1 + \frac{p(\tau)\beta}{((x+v\tau)^2 + H^2)^{\frac{\alpha}{2}}}\right) \mathrm d \tau \label{eq:maxobj}\\
\mathrm{s.t.} \;&\; \int_0^{\frac{y-x}{v}} p(\tau) \mathrm d \tau \le E.
\end{align}
\end{subequations}

Denote $s = x + v\tau$, and thus we have $\mathrm d s = v \mathrm d \tau$. By changing the variable from $\tau$ to $s$, the throughput maximization problem (\ref{prob:thrmax}) can be reformulated as
\begin{subequations}\label{prob:thrmax2}
\begin{align}
\max_{p(s) \ge 0} \;&\; \frac{1}{v} \int_{x}^{y} \frac{1}{2}W \log_2 \left( 1 + \frac{p(s)\beta}{(s^2 + H^2)^{\frac{\alpha}{2}}}\right) \mathrm d s \label{eq:thrmaxobj}\\
\mathrm{s.t.} \;&\; \frac{1}{v} \int_{x}^{y} p(s) \mathrm d s \le E. \label{eq:powerconstr}
\end{align}
\end{subequations}

Notice that the UAV receives the data from the sensor if and only if $p(s) > 0$. Otherwise, the UAV flies with the maximum speed. Therefore, the condition for $p(s) > 0$ needs to be specified. We have the following conclusion.

\begin{theorem} \label{thm:pt}
Let $p^*(s)$ be an optimal solution of the problem (\ref{prob:thrmax2}). It holds that $p^*(s) > 0$ for $x < s < y$ if and only if $x$, $y$ and $v$ satisfy
\begin{align}
(y-x) (\max\{x^2, y^2\} + H^2)^{\frac{\alpha}{2}} - \int_x^y (s^2 + H^2)^{\frac{\alpha}{2}} \mathrm d s \le \beta E v. \label{eq:cond}
\end{align}
Moreover, the optimal power allocation $p^*(s)$ is
\begin{align}
p^*(s) = \frac{1}{\gamma_0} - \frac{1}{\gamma(s)}, \label{eq:ps}
\end{align}
where the water level is
\begin{align}
\frac{1}{\gamma_0} &= \frac{vE}{y-x} + \frac{1}{(y-x)\beta} \int_x^y (s^2 + H^2)^{\frac{\alpha}{2}} \mathrm d s, \label{eq:gamma0}
\end{align}
and the inverse of channel gain is
\begin{align}
\frac{1}{\gamma(s)} &= \frac{(s^2 + H^2)^{\frac{\alpha}{2}}}{\beta}. \label{eq:gammat}
\end{align}

The corresponding optimal objective value of (\ref{eq:thrmaxobj}) is
\begin{align}
B_{\max}(x,y,v)
= \frac{W}{2v} \bigg(s \log_2\frac{\beta}{\gamma_0(s^2 + H^2)^{\frac{\alpha}{2}}} + \frac{\alpha s}{\ln2} - \frac{\alpha H}{\ln2} \arctan\frac{s}{H} \bigg) \bigg|_{s=x}^{s=y}. \label{eq:maxthr}
\end{align}
\end{theorem}

\begin{proof}
See Appendix \ref{proof:pt}.
\end{proof}

\begin{figure}[th]
\centering
\includegraphics[width=4.5in]{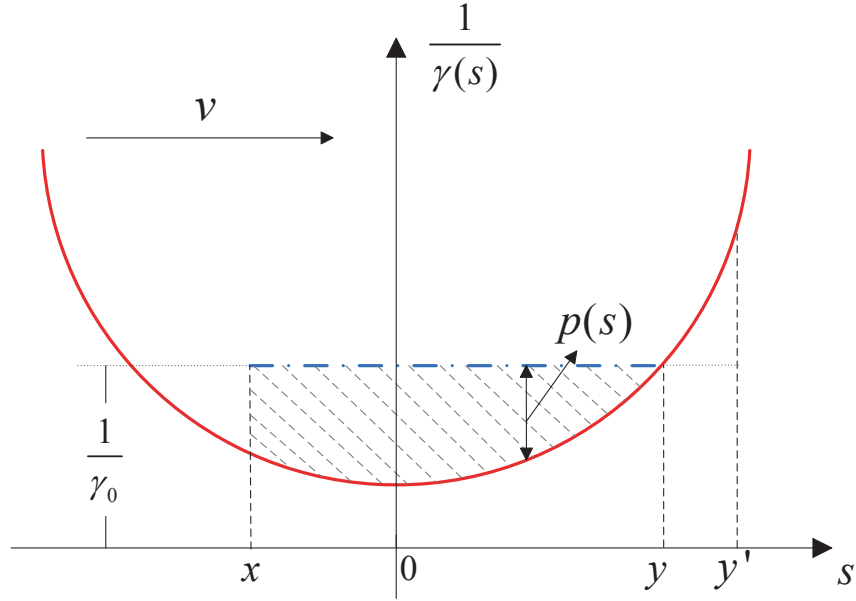}
\caption{Illustration of water-filling power allocation.} \label{fig:WF}
\end{figure}

The results in Theorem \ref{thm:pt} are interpreted in Fig.~\ref{fig:WF}. In this figure, the red solid curve represents the inverse of the channel gain $\frac{1}{\gamma(s)}$, the blue dash-dotted line represents the water level $\frac{1}{\gamma_0}$, and the area between the two curves represents the total energy budget. If $x,y$ and $v$ satisfy the condition (\ref{eq:cond}), we have $p(s)>0$ for $x < s < y$. However, if $x, y'$ and $v$ do not satisfy the condition (\ref{eq:cond}) as shown in the figure, there must be another set of $x, y$ and $v$ where $y < y'$ so that $p(s)>0$ for $x < s < y$. Therefore, the data upload must be within the range $[x, y]$. Theorem \ref{thm:pt} explicitly gives the feasible region of $x, y, v$ for optimal data collection.

For the free space LOS channel model with $\alpha = 2$, the condition can be further simplified by calculating the integration. In particular, we have
\begin{align}
\int_x^y (s^2 + H^2)^{\frac{\alpha}{2}} \mathrm d s \bigg|_{\alpha = 2}&= \int_x^y (s^2 + H^2) \mathrm d s \nonumber\\
&= \frac{y^3 - x^3}{3} + (y-x)H^2 \nonumber\\
&= (y-x) \bigg( \frac{x^2 + xy + y^2}{3} + H^2 \bigg).
\end{align}
Replacing the term $\int_x^y (s^2 + H^2)^{\frac{\alpha}{2}} \mathrm d s$ by the above expression, the condition (\ref{eq:cond}) can be expressed as the following two conditions:
\begin{itemize}
\item[] (a) $|x| \le |y|$ and ${2y^3 + x^3 - 3y^2x} \le {3\beta E}v$,

\item[] (b) $|x| \ge |y|$ and ${3x^2y - 2x^3 - y^3} \le {3\beta E}v$,
\end{itemize}
and the water level can be expressed as
\begin{align}
\frac{1}{\gamma_0} &= \frac{vE}{y-x} + \frac{x^2 + xy + y^2}{3\beta} + \frac{H^2}{\beta}. \label{eq:gamma0f}
\end{align}

\subsection{UAV's Speed Optimization} \label{sec:speed}
With the optimal power allocation, the maximum throughput $B_{\max}(x,y,v)$ in (\ref{eq:maxthr}) has the following property.

\begin{theorem} \label{prop:decv}
The maximum throughput $B_{\max}(x,y,v)$ is a decreasing function of $v$.
\end{theorem}
\begin{proof}
See Appendix \ref{proof:decv}.
\end{proof}

Based on Theorem \ref{prop:decv}, the feasibility of any solution $(x, y, v, p(t))$ is guaranteed if the minimum speed satisfies (\ref{eq:cond}). The minimum speed can be written as a function of $x$ and $y$, i.e.,
\begin{align}
v_{\mathrm{m}}(x, y) = \min \bigg\{ v_{\max}, \frac{1}{\beta E} \Big( (y-x) (\max\{x^2, y^2\} + H^2)^{\frac{\alpha}{2}} - \int_x^y (s^2 + H^2)^{\frac{\alpha}{2}} \mathrm d s \Big)\bigg\}. \label{eq:vm}
\end{align}
When $\alpha = 2$, the minimum speed can be rewritten as
\begin{align}
v_{\mathrm{m}}(x,y)
= \left\{ \begin{array}{ll} \frac{2y^3 + x^3 - 3y^2x} {3\beta E}, & \textrm{if~} |x| \le |y| \textrm{~and~}  2y^3 + x^3 -3y^2x \le 3\beta E v_{\max},\\
\frac{3x^2y - 2x^3 - y^3} {3\beta E}, & \textrm{if~} |x| \ge |y| \textrm{~and~}  3x^2y - 2x^3 - y^3 \le 3\beta E v_{\max},\\
v_{\max}, & \textrm{otherwise}.
\end{array} \right. \label{eq:vmin}
\end{align}

According to Theorems \ref{thm:pt} and \ref{prop:decv}, for given $x, y$ and the optimal power allocation, the optimization over $v$ can be formulated as
\begin{subequations}\label{prob:single_v}
\begin{align}
\min_{v} \;&\; \frac{(S_2-S_0)}{v_{\max}} +  (y - x)\Big(\frac{1}{v} - \frac{1}{v_{\max}}\Big) \label{eq:tmin_v}\\
\mathrm{s.t.} \;&\; B_{\max}(x,y, v) \ge B, \label{eq:vfeas}\\
\;&\; v_{\mathrm{m}}(x,y) \le v \le v_{\max}. \label{eq:vrange}
\end{align}
\end{subequations}

Problem (\ref{prob:single_v}) can be solved in two steps. Firstly, we check the feasibility of problem (\ref{prob:single_v}). Based on Theorem \ref{prop:decv}, if $B_{\max}(x,y, v_{\mathrm{m}}(x,y)) \ge B$, (\ref{prob:single_v}) is feasible, and we go to the second step. As the objective function (\ref{eq:tmin_v}) is a decreasing function of $v$, the optimal speed, denoted by $v^*(x,y)$, is the maximum feasible speed that satisfies (\ref{eq:vfeas}) in $[v_{\mathrm{m}}(x,y), v_{\max}]$. Since $B_{\max}(x,y,v)$ is a decreasing function of $v$, $v^*(x,y)$ can be found by bisection search algorithm. In summary, the algorithm to obtain the optimal $v$ and $p(s)$ for given $x$ and $y$ where $x < y$ in problem (\ref{prob:single_nohover}) is summarized in Algorithm \ref{alg:vp}.

{
\renewcommand{\algorithmicrequire}{\textbf{Input:}}
\renewcommand{\algorithmicensure}{\textbf{Output:}}

\begin{algorithm}[th]
\caption{Calculate $v$ and $p(s)$ for given $(x, y)$ in problem (\ref{prob:single_nohover})} \label{alg:vp}
\begin{algorithmic}[1]

\REQUIRE {$\beta, H, W, \alpha, B, E, \delta, x$ and $y$ where $x < y$.}

\ENSURE {$v^*, p^*(s)$.}

\STATE Calculate $v_{\mathrm{m}}(x,y)$ according to (\ref{eq:vm}).

\IF {$B_{\max}(x,y, v_{\mathrm{m}}(x,y)) \ge B$}

\STATE Find the maximum $v \in [v_{\mathrm{m}}(x,y), v_{\max}]$ that satisfies $B_{\max}(x,y, v) \ge B$ by bisection search.

\STATE Set $v^* = v$, and $p^*(s)$ is calculated according to (\ref{eq:ps})-(\ref{eq:gammat}).

\ELSE

\STATE Problem (\ref{prob:single_nohover}) for given $x$ and $y$ is infeasible.

\ENDIF

\end{algorithmic}
\end{algorithm}

In this algorithm, line 2 examines the feasibility of the problem. If the inequality does not hold, there is no feasible solution for the given parameters, and the algorithm terminates. Otherwise, bisection search is launched as in line 3. When the bisection search terminates, the optimal solution is recorded in line 4.

}

\begin{remark}
It is interesting to remark that the maximum throughput in (\ref{prob:thrmax2}) can be re-written as
\begin{align}
B_{\max}(x,y,v) = E\max_{p(s)} \frac{ \frac{1}{y-x} \int_{x}^{y} \frac{1}{2} W \log_2 \left( 1 + \frac{p(s)\beta}{(s^2 + H^2)^{ \frac{\alpha}{2}}}\right) \mathrm d s}{\frac{v}{y-x}E},
\end{align}
where $p(s)$ is constrained by $\frac{1}{y-x}\int_{x}^{y} p(s) \mathrm d s  \le \frac{v}{y-x}E$ which should be satisfied with equality to achieve the maximum. The term on the right side of the operator $\max$ can be viewed as the energy efficiency (achievable data rate per unit power) with ``average power budget" $\frac{v}{y-x}E$. Therefore, Theorem \ref{prop:decv} says that the energy efficiency is a decreasing function of the power budget in fading channels. It extends the result from the AWGN channel \cite{verdu1990on} to the UAV LOS channel.
\end{remark}

\subsection{Data Collection Interval Optimization}
Finally, we consider the problem of determining $x$ and $y$ in problem (\ref{prob:single_nohover}), which can be written as
\begin{subequations}\label{prob:single_xy}
\begin{align}
\min_{x, y} \;&\; \frac{(S_2-S_0)}{v_{\max}} +  (y - x)\Big(\frac{1}{v^*(x, y)} - \frac{1}{v_{\max}}\Big) \label{eq:tmin_xy}\\
\mathrm{s.t.} \;&\; S_0 \le x < y \le S_2,
\end{align}
\end{subequations}
where $v^*(x,y)$ is the optimal solution of (\ref{prob:single_v}). As $v^*(x, y)$ is a complex function of $x$ and $y$, there is no efficient algorithms other than two-dimensional line search to solve the problem (\ref{prob:single_xy}). By sampling $m$ points in the range $[S_0, S_2]$ with identical inter-point distance, the total number of search pairs $(x, y)$ where $x < y$ is $\frac{m(m+1)}{2}$. Thus, the complexity of the two-dimensional search is $O(m^2)$.

\section{Flight Time Minimization for Multi-sensor Case} \label{sec:multi}
In the multi-sensor case, the data upload intervals for the sensors correlates with one another. In particular, if a sensor's data upload interval is wide, the one next to it can only have a short data upload interval. To deal with the inter-sensor correlation, we adopt the DP approach \cite{bertsekas2005dynamic} to solve the flight time minimization problem for multiple sensors. Firstly, the basic concept of the DP algorithm is briefly reviewed as follows.

{\subsection{Brief Introduction to DP Algorithm}
The DP algorithm deals with decision making problems in dynamic systems which can be divided into \emph{stages}. The dynamic system expresses the evolution of the system \emph{states} $s_k \in \mathcal{S}_k$ under the influence of state-dependent control \emph{actions} $u_k \in \mathcal{U}_k(s_k)$ taken at discrete instances of time (\emph{stage}) $k$. The system state updates as $s_{k+1} = f_k(s_k, u_k)$. There is an additive \emph{cost} $g_k(s_k, u_k)$ in each stage. The objective is to minimize the total cost by finding the optimal control actions for a given initial state, i.e., $\min_{u_0, u_1, \cdots, u_{K-1}} \left\{ \sum_{k=0}^{K-1} g_k(s_k, u_k) + g_K(s_K) | s_0\right\}.$ This problem can be solved by the DP algorithm \cite[Prop.~1.3.1, Vol.~I]{bertsekas2005dynamic}, i.e., proceeding the following backward in time from stage $K-1$ to stage 0
\begin{align}
&J_K(s_K) = g_K(s_K), \quad \forall s_K \in \mathcal{S}_K, \nonumber\\
&J_k(s_k) = \min_{u_k \in \mathcal{U}_k(s_k)} \left\{ g_k(s_k, u_k) + J_{k+1}(f_k(s_k, u_k)) \big| s_k\right\}, \quad \forall s_k \in \mathcal{S}_k, k = K-1, \cdots, 0, \nonumber
\end{align}
where $J_k(s_k)$ is termed as the optimal \emph{cost-to-go}, i.e., the minimum cost for the ($K-k$)-stage problem that starts at stage $k$ with state $s_k$ and ends at stage $K$.
}

\subsection{DP-based Flight Time Minimization}
Now, we apply the DP algorithm to solve the flight time minimization problem (\ref{prob:multisensor}). Firstly, the objective function (\ref{eq:tmin}) can be rewritten as
\begin{align}
&\; \min_{\bm x, \bm y, \bm v, \bm p(t)} \frac{S_{N+1}-S_0}{v_{\max}} + \sum_{n=1}^N \left(t_n -\frac{y_n-x_n}{v_{\max}} \right) \nonumber\\
=& \quad \min_{\bm x, \bm y} \; \bigg[ \frac{S_{N+1}-S_0}{v_{\max}} + \sum_{n = 1}^N \min_{v_n, p_n(t)}\left(t_n -\frac{y_n-x_n}{v_{\max}} \right)\bigg], \label{eq:probTmin}
\end{align}
where the minimization over $v_n, p_n(t)$ for a given pair $x_n < y_n$ corresponds to the single-sensor flying case (\ref{prob:single_nohover}) and can be efficiently solved by Algorithm \ref{alg:vp}. If $x_n = y_n$, i.e., the UAV hovers at location $x_n$, the minimization takes the value with $t_n = T_{\mathrm{h},n}(x_n)$ as the solution for (\ref{eq:hovereq}). Thus, according to the results in Sections \ref{sec:hover}), \ref{sec:pa}, and \ref{sec:speed}, we can define a cost function as
\begin{align}
g_n(x_n, y_n) &= \min_{v_n, p_n(t)}\left(t_n -\frac{y_n-x_n}{v_{\max}} \right)\nonumber\\
&= \left\{\!\begin{array}{ll}
T_{\textrm{h},n}(x_n), & \textrm{if~} x_n = y_n,\\
(y_n \!-\! x_n)\big(\frac{1}{v^*_{n}(\tilde x_n, \tilde y_n)} \!- \! \frac{1}{v_{\max}}\big), &\textrm{if~} x_n < y_n \textrm{~and~} (\ref{prob:single_v}) \textrm{~is~feasible},\\
+\infty, & \textrm{elsewhere},
\end{array}\right. \label{eq:cost}
\end{align}
for all $n = 1, 2, \cdots, N$, where $\tilde x_n = x_n - S_n, \tilde y_n = y_n - S_n$ are the horizontal coordinates relative to $S_n$, $v^*_{n}(\tilde x_n, \tilde y_n)$ is the optimal feasible solution of (\ref{prob:single_v}) that can be calculated via Algorithm \ref{alg:vp}, and $T_{\textrm{h},n}(x_n)$ is the minimum hovering time obtained by solving (\ref{eq:hovereq}). If $x_n < y_n$ while (\ref{prob:single_v}) is infeasible, we set the cost as infinity.

Based on the above cost function, we formulate the flight time minimization problem as a DP problem. In particular, we have
\begin{itemize}
\item[-] index of stage: $n$,
\item[-] system state in stage $n$: the end point of data upload for sensor $n-1$, denoted by $s_n = y_{n-1}$. The state space is $\mathcal{S}_n = [S_0, S_{N+1}]$,
\item[-] control action in stage $n$: the data upload interval for sensor $n$, i.e., $(x_n, y_n)$. The action space is $\mathcal{U}_n(s_n) = \{(x_n, y_n)|s_n \le x_n \le y_n \le S_{N+1}\}$,
\item[-] state update rule: $s_{n+1} = f_n(s_n, x_n, y_n) = y_n$,
\item[-] per-stage cost: $g_n(x_n, y_n), n = 1, 2, \cdots, N$ as defined in (\ref{eq:cost}), and $g_{N+1}(s_{N+1}) = \frac{S_{N+1}-S_0}{v_{\max}}$.
\end{itemize}
As a result, the problem (\ref{eq:probTmin}) can be rewritten as
\begin{align}
\min_{\bm x, \bm y} \left[ \sum_{n = 1}^N g_n(x_n, y_n) + g_{N+1}(y_{N}) \right],
\end{align}
which can be solved by recursively calculating the cost-to-go function stage-by-stage as
\begin{align}
&J_{N+1}(s_{N+1}) = g_{N+1}(s_{N+1}) = \frac{S_{N+1}-S_0}{v_{\max}}, \quad \forall s_{N+1}, \label{eq:DPN}\\
&J_n(s_n) = \min_{s_n \le x_n \le y_n \le S_{N+1}} \{g_n(x_n, y_n) + J_{n+1}(y_n)\}, \quad \forall s_{n}, n = N, N-1, \cdots, 1. \label{eq:DPn}
\end{align}
Then, the minimum flight time can be obtain in the last step, i.e.
\begin{align}
T_{\min} =  J_1(S_0).
\end{align}

In addition, if the optimal control actions for (\ref{eq:DPn}) are $(x_1^*, y_1^*), (x_2^*, y_2^*), \cdots, (x_N^*, y_N^*)$, the optimal solution for the problem (\ref{eq:probTmin}) is $\bm x^* = \{x_1^*, x_2^*, \cdots, x_N^*\}, \bm y^* = \{y_1^*, y_2^*, \cdots, y_N^*\}$. Thus, the optimal solution of the original problem (\ref{prob:multisensor}) is $\bm x^*, \bm y^*$ joint with the optimal speeds $v^*_{n}, n = 1, 2, \cdots, N$ obtained by solving (\ref{prob:single_v}) and the optimal power allocation in (\ref{eq:ps}).

It is remarkable that the computational complexity for the calculation of the cost-to-go functions $J_n(s_n)$ can be reduced by exploring the property of (\ref{eq:DPn}).

\begin{proposition} \label{prop:complex}
Concerning the DP algorithm (\ref{eq:DPN}) and (\ref{eq:DPn}), for any given $n$ and $s_n$, if the optimal solution $(x_n^*, y_n^*)$ for the minimization problem in (\ref{eq:DPn}) satisfies $x_n^* > s_n$, we have $J_n(s_n') = J_n(s_n)$ for all $s_n' \in [s_n, x_n^*]$.
\end{proposition}
\begin{proof}
See Appendix \ref{proof:complex}.
\end{proof}

According to Proposition \ref{prop:complex}, to reduce the computational complexity, the calculation of $J_n(s_n)$ for a given $n$ can be launched from the initial point $S_0$ to the destination $S_{N+1}$. When an optimal solution $(x_n^*, y_n^*)$ for a given $s_n$ is found and satisfies $s_n < x_n^*$, the calculation of $J_n(s_n')$ for $s_n' \in [s_n, x_n^*]$ can be omitted as the optimal solutions are equivalent to $(x_n^*, y_n^*)$.

\section{Numerical Results} \label{sec:sim}
Some numerical results are shown in this section. In the numerical simulations, we set $H = 100$ m, {the reference SNR at transmission distance 1 m is set to} $\beta = 80$ dB \cite{lyu2016cyclical}, and the channel bandwidth $W = 20$ kHz. According to the state-of-the-art in the industry \cite{Dji2017inspire}, we set the maximum speed $v_{\max} = 26$ m/s.

\subsection{Single-sensor Case Study}
The optimal result for the single-sensor case versus different values of data upload requirement $B$ and sensor energy constraint $E$ with $S_0 = -5000$ m, $S_1 = 0$ m, $S_2 = 5000$ m are depicted in Figs.~\ref{fig:xversusB} and \ref{fig:xversusE}, respectively. It can be seen that the optimal transmission interval $(x, y)$ is symmetric, which corresponds to the shortest average transmission distance from the sensor to the UAV. In Fig.~\ref{fig:xversusB}, when $B > 5.7$ Mb, the optimal solution is hovering above the sensor to receive data. When $2.5$ Mb $ < B < 5.7$ Mb, both the length of data upload interval and the UAV's speed decreases as the data upload requirement increases. The decrease of data upload interval increases the channel gain between the UAV and the sensor, and the decrease of the UAV's speed increases the transmission time. When $B < 2.5$ Mb, the UAV can fly with the maximum speed while successfully receive all the uploaded data. In this range, the minimum data upload interval is depicted, and its length decreases as the data upload requirement decreases as the time required for data upload decreases.

\begin{figure}
\centering
\includegraphics[width=4.4in]{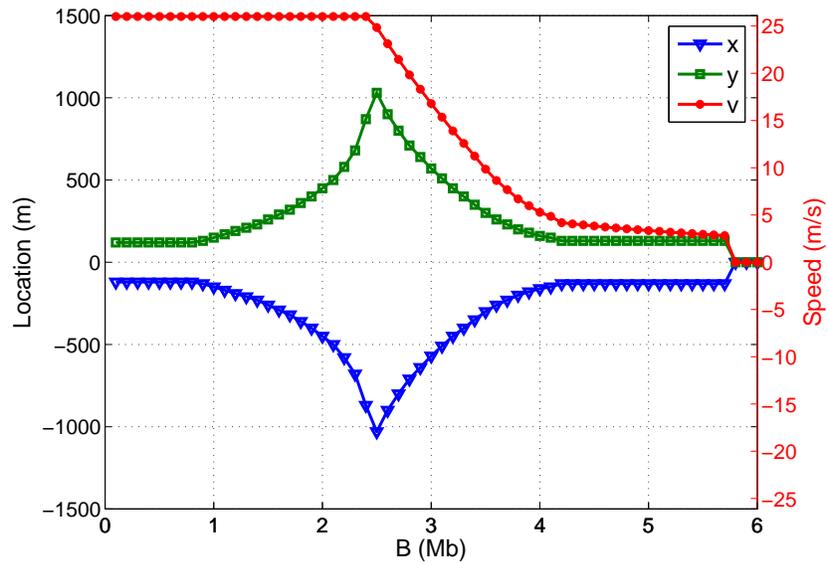}
\caption{Optimal solution $(x, y, v)$ versus $B$ for the single-sensor case, with $S_0 = -5000$ m, $S_1 = 0$ m, $S_2 = 5000$ m, and $E = 1$ J.} \label{fig:xversusB}
\end{figure}

\begin{figure}
\centering
\includegraphics[width=4.4in]{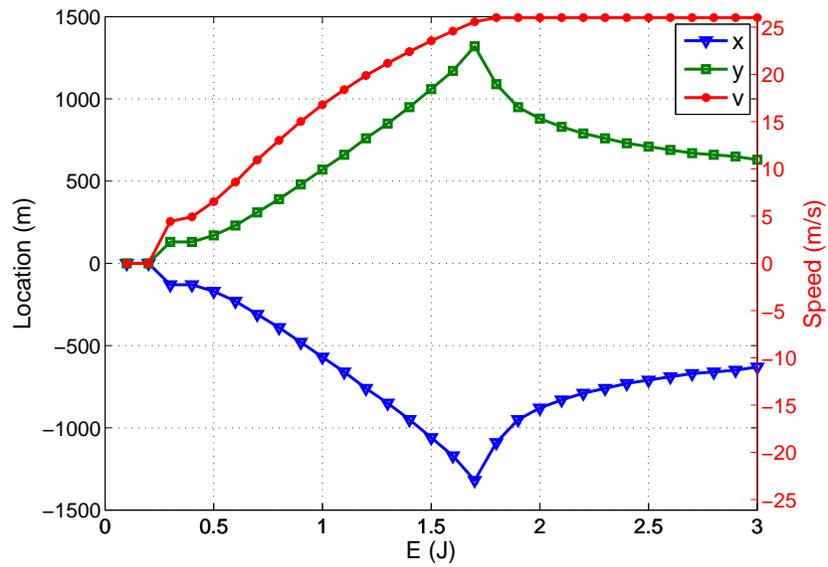}
\caption{Optimal solution $(x, y, v)$ versus $E$ for the single-sensor case, with $S_0 = -5000$ m, $S_1 = 0$ m, $S_2 = 5000$ m, and $B = 3$ Mb.} \label{fig:xversusE}
\end{figure}

In Fig.~\ref{fig:xversusE}, the optimal result versus $E$ is opposite to that versus $B$. In particular, when $E < 0.3$ J, the UAV also needs to hover above the sensor to receive data. When $0.3$ J $< E < 1.7$ J, both the length of the data upload interval and the UAV's speed increases as the amount of energy increases. While for $E > 1.7$ J, the UAV can fly with the maximum speed, and the minimum length of the data upload interval decreases as the amount of energy increases.

\subsection{Multi-sensor Case Study}
Then, the data collection for multiple sensors is studied by simulation. In particular, the UAV flies from $S_0 = 0$ m to $S_{N+1} = 10000$ m, during which $N=10$ sensors are deployed. The locations of the sensors $S_n, n = 1, \cdots, 10$ are fixed as 500m, 2500m, 4500m, 6500m, 7000m, 7500m, 8000m, 8500m, 9000m, and 9500m, i.e., the first four sensors are 2000m apart from one another (sparsely deployed), and the last six sensors are 500m apart from one another (densely deployed). We study the impact of required data and energy limitation respectively.

\begin{figure}
\centering
\includegraphics[width=4.4in]{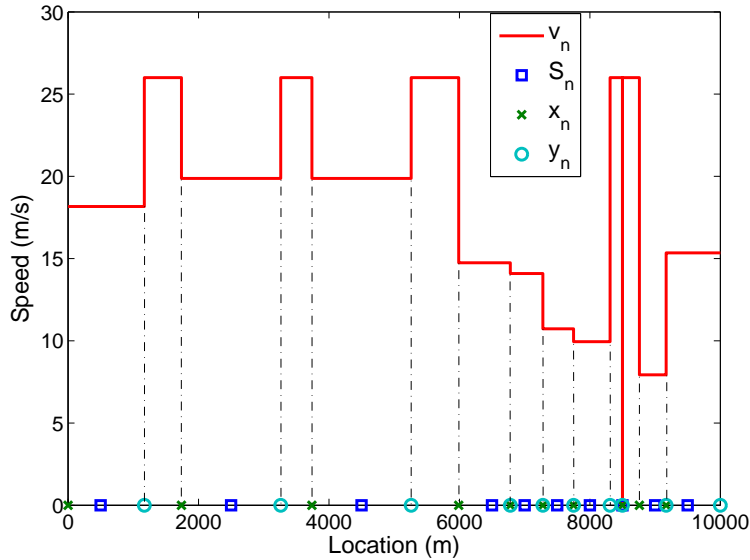}
\caption{Optimal solution $(x_n, y_n, v_n)$ for $N = 10$ sensors with $E_n = 1.2$ J for all $n = 1, 2, \cdots, N$, $B_1 = \cdots = B_4 = B_6 = B_{10} = 3$ Mbits, $B_5 = 2.5$ Mbits, $B_7 = B_{9} = 3.5$ Mbits, and $B_8 = 7$ Mbits.} \label{fig:OptvsB}
\end{figure}

In Fig.~\ref{fig:OptvsB}, the amount of energy in each sensor is set identical, $E_n = 1.2$ J for all $n = 1, 2, \cdots, N$, and the amount of data to be transmitted varies. We set $B_1 = \cdots = B_4 = B_6 = B_{10} = 3$ Mbits, $B_5 = 2.5$ Mbits, $B_7 = B_{9} = 3.5$ Mbits, and $B_8 = 7$ Mbits. It can be seen that as the first four sensors are sparsely located, the upload intervals are disconnected. The reason is that it is not energy-efficient when the transmission distance is large. In this case, the UAV collects data from a sensor in a short range and then flies towards another with the maximum speed. For the last six sensors, as the amount of data to be transmitted increases from sensor $S_5$ to sensor $S_8$ and then decreases from $S_8$ to $S_{10}$, the UAV's speed firstly decreases and then increases accordingly so that the required data can be uploaded successfully. Particularly, As the amount of data in sensor $S_8$ is extremely large, the UAV hovers above it to collect data with maximum data rate, so that the overall flight time is minimized.

\begin{figure}
\centering
\includegraphics[width=4.4in]{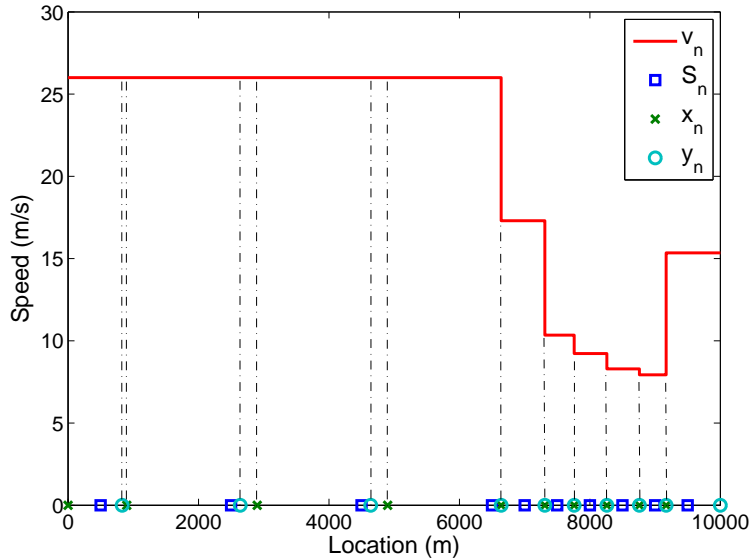}
\caption{Optimal solution $(x_n, y_n, v_n)$ for $N = 10$ sensors with $E_n = 1.2$ J for all $n = 1, 2, \cdots, N$, $B_1 = \cdots = B_4 = B_6 = B_{10} = 2$ Mbits, $B_5 = 2.5$ Mbits, $B_7 = B_{9} = 3.5$ Mbits, and $B_8 = 3.8$ Mbits.} \label{fig:OptvsBv2}
\end{figure}

In Fig.~\ref{fig:OptvsBv2}, we change the values of the amount of data bits as $B_1 = \cdots = B_4 = B_6 = B_{10} = 2$ Mbits, $B_5 = 2.5$ Mbits, $B_7 = B_{9} = 3.5$ Mbits, and $B_8 = 3.8$ Mbits. Firstly, as the data bits in sensors $S_1, \cdots, S_4$ are limited, the UAV can successfully receive all the data bits when flying with maximum speed. In addition, as the data bits in sensor $S_8$ are reduced compared with Fig.~\ref{fig:OptvsB}, the hovering mode is not necessary any more. As the amount of data bits is still the largest, the speed is quite low.

\begin{figure}
\centering
\includegraphics[width=4.4in]{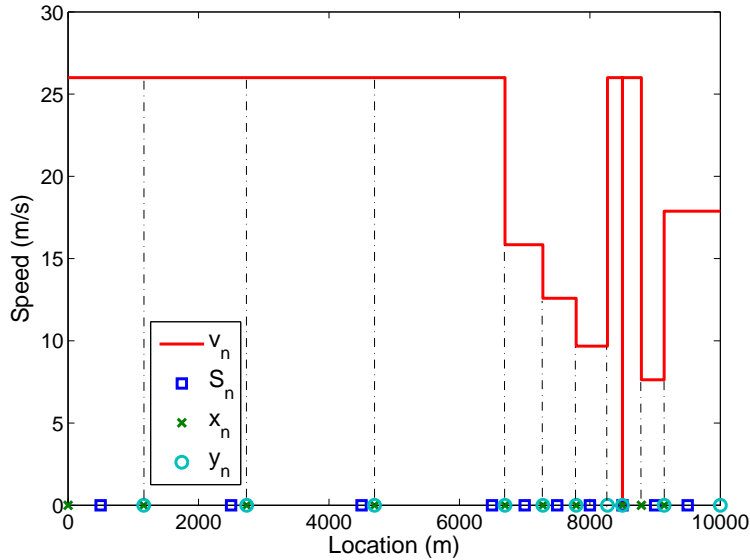}
\caption{Optimal solution $(x_n, y_n, v_n)$ for $N = 10$ sensors with $B_n = 3$ Mbits for all $n = 1, 2, \cdots, N$, and $E_1 = \cdots = E_4 = 3.6$ J, $E_5 = 3.2$ J, $E_6 = E_{10} = 1.8$ J, $E_7 = E_9 = 0.8$ J, and $E_8 = 0.2$ J.} \label{fig:OptvsE}
\end{figure}

Then, we set the amount of data in each sensor to be the same, i.e., $B_n = 3$ Mbits for all $n = 1, 2, \cdots, N$, while $E_1 = \cdots = E_4 = 3.6$ J, $E_5 = 3.2$ J, $E_6 = E_{10} = 1.8$ J, $E_7 = E_9 = 0.8$ J, and $E_8 = 0.2$ J to evaluate the impact of the energy constraint. It can be seen that with sufficient amount of energy for the first four sensors, the UAV can fly with maximum speed while successfully receiving all the data. For the last six sensors, as the amount of energy firstly decreases and then increases from sensor $S_5$ to sensor $S_{10}$, the optimal speed also decreases at first and then increase. In particular, as the eighth sensor is quite energy stringent, the UAV hovers above it with zero speed to collect its data. In addition, the transmission intervals for the first four sensors shift towards the initial point so that more space can be reserved for the last six sensors which has limited energy budget.

\begin{figure}
\centering
\includegraphics[width=4.4in]{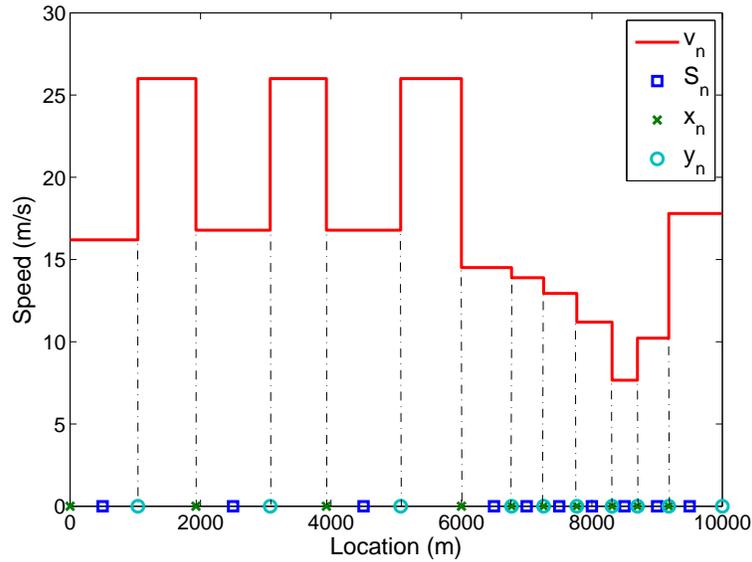}
\caption{Optimal solution $(x_n, y_n, v_n)$ for $N = 10$ sensors with $B_n = 3$ Mbits for all $n = 1, 2, \cdots, N$, and $E_1 = E_2 = E_3 = E_7 = E_9 = 1.0$ J, $E_4 = 1.2$ J, $E_5 = 3.2$ J, $E_6 = E_{10} = 2.0$ J, and $E_8 = 0.6$ J.} \label{fig:OptvsEv2}
\end{figure}

Next, we reset the amount of energy as $E_1 = E_2 = E_3 = E_7 = E_9 = 1.0$ J, $E_4 = 1.2$ J, $E_5 = 3.2$ J, $E_6 = E_{10} = 2.0$ J, and $E_8 = 0.6$ J and re-run the simulation, the result is shown in Fig.~\ref{fig:OptvsEv2}. It can be found that the UAV serves the first three sensors with medium speed, as the limited amount of energy cannot support the maximum speed. Similarly, as the energy in sensor $S_8$ is sufficient to support data collection during flying, the hovering mode is not necessary. Compared with Figs.~\ref{fig:OptvsB}-\ref{fig:OptvsBv2}, the energy constraint has similar impact as the data requirement.

\begin{figure}
\centering
\includegraphics[width=4.4in]{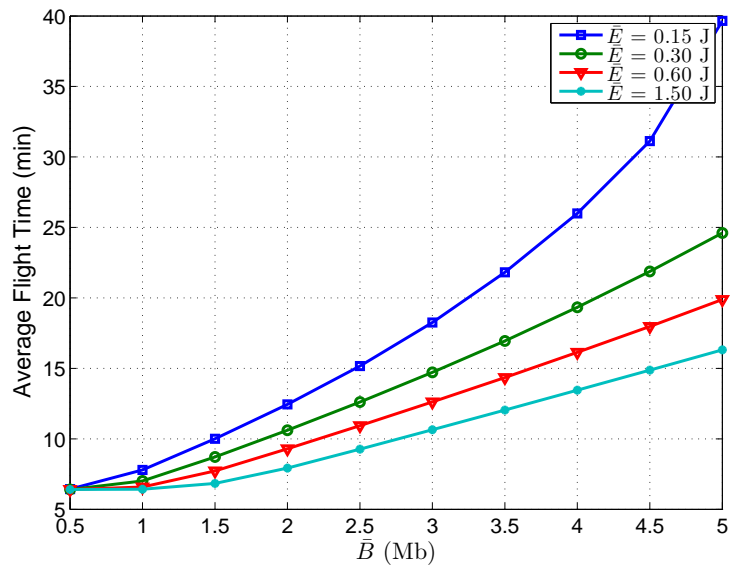}
\caption{Average flight time versus average amount of data with random data requirement, random energy and random locations.} \label{fig:AveTimevsB}
\end{figure}

\subsection{Average Performance Evaluation}
We further evaluate the average performance with random data requirement, random energy and random locations. The amount of data in each sensor follows uniform distribution with a mean value $\bar{B}$, the amount of energy in each sensor follows uniform distribution with a mean value $\bar{E}$, and each pair ($B_n, E_n$) is set to satisfy the feasibility constraint (\ref{cond:feas}). The sensors are uniformly distributed in the range $[S_0, S_{N+1}] = [0, 10000]$ m. The results are illustrated in Figs.~\ref{fig:AveTimevsB} and \ref{fig:AveTimevsE}. It can be seen in Fig.~\ref{fig:AveTimevsB} that when the average amount of energy is sufficient, the average flight time grows almost linearly with the increase of $\bar{B}$. But when the amount of energy is deficient (e.g., $\bar{E} = 0.15$ J), the average flight time grows exponentially with the increase of $\bar{B}$. This is due to the different relations between the flight time and the amount of data in hovering mode and flying mode. In energy sufficient case, the UAV collects data mainly in flying mode. While in energy constrained case, it collects data mainly in hovering mode.

\begin{figure}
\centering
\includegraphics[width=4.4in]{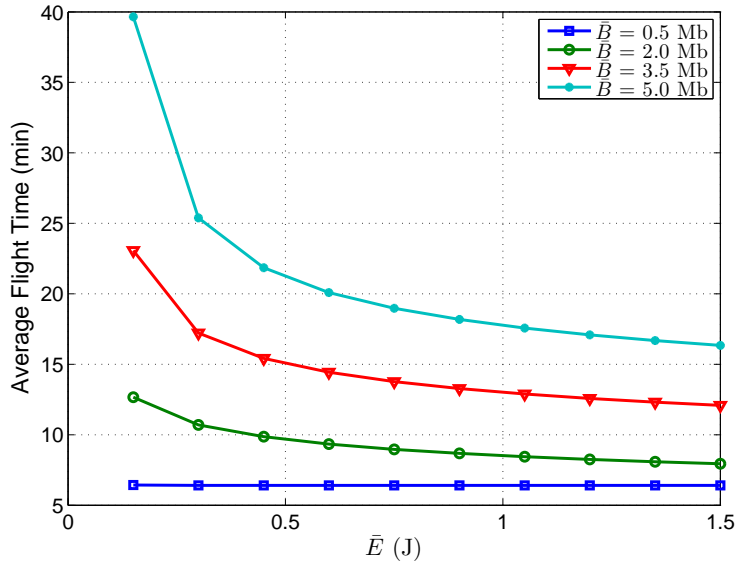}
\caption{Average flight time versus average amount of energy with random data requirement, random energy and random locations.} \label{fig:AveTimevsE}
\end{figure}

In Fig.~\ref{fig:AveTimevsE}, it is observed that when the amount of data is small, the average flight time is constant over all examined values of $\bar{E}$, which means that the UAV can fly with the maximum speed and collect data during flying. In addition, it is expected with the increase of $\bar{E}$, the curves converges to a fixed point with minimum flight time, i.e., the UAV flies with the maximum speed. However, the figure shows that the convergence is slow, especially for large values of $\bar{B}$. For the case with $\bar{B} = 5.0$ Mb, the curve firstly goes down exponentially, and then linearly with close-to-zero slope.

\begin{figure}
\centering
\includegraphics[width=4.4in]{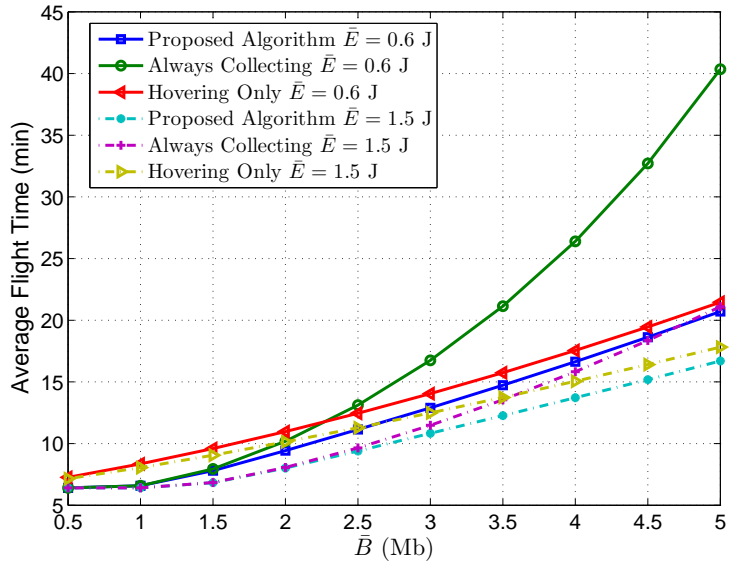}
\caption{Performance comparison with random data requirement, random energy and random locations.} \label{fig:AveTimeCMP}
\end{figure}

{Finally, we compare the proposed algorithm with the following two baselines: (1) \emph{hovering only} policy, i.e., the UAV flies with the maximum speed to the top of each sensor and hovers for data collection; (2) \emph{always collecting} algorithm developed based on \cite{lyu2016cyclical}, i.e., the whole range is divided by $S_0 = z_0 \le z_1 \le \cdots \le z_{N-1} \le z_N = S_{N+1}$, and the sensor $n$ uploads data during interval $[z_{n-1}, z_n]$ with the constant power $p_n(t) = v_nE_n/(z_n - z_{n-1})$. The transmission intervals are optimized in a way similar to our proposed DP algorithm. Different from the always collecting algorithm, our algorithm allows non-consecutive intervals so that the UAV can fly over poor channel regions with the maximum speed to save time. The result is shown in Fig.~\ref{fig:AveTimeCMP}. It can be seen that given an average energy, there is a performance gap between the proposed algorithm and the hovering only policy. Therefore, collecting data during flight can save some time even if the data rate is lower than hovering mode. Concerning the always collecting algorithm, it performs almost the same as the proposed algorithm when the average data amount is small, but the performance gap increases when the data amount becomes large. The reason is that in the always collecting algorithm, the average data rate is lower due to larger distance compared with the proposed algorithm. Such a low data rate can support a high speed for small data amount, but not for large data amount in which case the UAV has to fly with very low speed to guarantee successful data collection.
}

\section{Conclusion} \label{sec:concl}
In this paper, we have solved the flight time minimization problem for completing the data collection mission in a one-dimensional sensor network. The analysis on hovering mode provides the feasibility condition for a successful data collection. The analysis on the single-sensor case reveals the optimal solution structures. Firstly, the optimal power allocation follows the classical water-filling policy. Secondly, the maximum amount of data bits that can be successfully uploaded during UAV's flying is a decreasing function of the UAV's speed, which results in a simple bisection method to find the optimal speed. For the multi-sensor case, we have shown that the division of data collection intervals can be optimized via the DP algorithm. According to the numerical results, is has been observed that the behavior of the UAV relies on the locations, the data amount and the energy amount of sensors. With a sufficient amount of energy, the UAV can fly with maximum speed. Otherwise, its speed is proportional to the sensors' energy budgets and the inter-sensor distance, but inversely proportional to the amount of data to be uploaded.

{ One possible extension of this work would be sensors' visit order optimization when the sensors are deployed on a 2D space. Also, considering sensors' mobility and channel fading would be an interesting direction of future research.}

\appendices

\section{Proof of Lemma \ref{lemma:inc}} \label{proof:inc}
Since
\begin{align}
f''(x) = -\frac{a^2}{x(x+a)^2} < 0,
\end{align}
$f'(x)$ is decreasing. Therefore,
\begin{align}
f'(x) = \log_2 \Big( 1+ \frac{a}{x} \Big) - \frac{a}{x+a} > f'(+\infty) = 0,
\end{align}
which indicates that $f(x)$ is increasing. In addition,
\begin{align}
\lim_{x \rightarrow +\infty} f(x) = \lim_{x \rightarrow +\infty} a \log_2 \Big( 1+ \frac{a}{x} \Big)^{\frac{x}{a}} = a \log_2 e.
\end{align}
Hence, $f(x) < a \log_2 e = \frac{a}{\ln2}$.

\section{Proof of Proposition \ref{prop:feas}} \label{proof:feas}
As
\begin{align}
&\;\frac{T_{\mathrm h, n}(x_n)}{2} W\log_2 \bigg( 1+ \frac{\beta E_n}{T_{\mathrm h, n}(x_n) ((x_n - S_n)^2 + H^2)^{\frac{\alpha}{2}}}\bigg) \nonumber\\
<&\; \frac{W\beta E_n}{2((x_n - S_n)^2 + H^2)^{\frac{\alpha}{2}} \ln 2} \nonumber \\
\le&\; \frac{W\beta E_n}{2H^{\alpha} \ln 2},
\end{align}
where the first inequality holds according to Lemma \ref{lemma:inc}, and the gap can be arbitrarily small as $T_{\mathrm h, n}(x_n)$ tends to infinity. In the second inequality, the equality holds when $x_n = S_n$. Hence, if $B_n < \frac{W\beta E_n}{2H^{\alpha} \ln 2}$, there is always a feasible transmission mode so that $B_n$ bits can be successfully transmitted.

If $B_n \ge \frac{W\beta E_n}{2H^{\alpha} \ln 2}$ on the contrary, the left hand side of (\ref{eq:hover}) is always less than $B_n$. Therefore, the equation (\ref{eq:hovereq}) is not feasible.

\section{Proof of Theorem \ref{thm:pt}} \label{proof:pt}
The Lagrangian function of the problem (\ref{prob:thrmax2}) is expressed as
\begin{align}
\mathcal{L} = \frac{W}{2v} \int_{x}^{y} \log_2 \left( 1 + \frac{p(s)\beta}{(s^2 + H^2)^{\frac{\alpha}{2}}}\right) \mathrm d s - \lambda \bigg(\frac{1}{v} \int_{x}^{y} p(s) \mathrm d s - E \bigg).
\end{align}
By setting $\frac{\partial \mathcal L}{\partial p(s)} = 0$, we get the optimal power allocation expressed as (\ref{eq:ps}), where $\gamma_0 = \frac{1}{\lambda}$ is the water level so that (\ref{eq:powerconstr}) is satisfied with equality, and the channel gain is
\begin{align}
\gamma(s) = \frac{\beta}{(s^2 + H^2)^{\frac{\alpha}{2}}}, \label{eq:gammas}
\end{align}
which is equivalent to (\ref{eq:gammat}).

Based on (\ref{eq:ps}) and (\ref{eq:gammas}), it can be found that $p^*(s) > 0$ must hold in a continuous interval. Next, we derive the necessary and sufficient condition for $p^*(s) > 0$ for all $x < s < y$.

\emph{1) Sufficiency:}

Suppose that $p^*(s) > 0$ for $x < s < y$. As $s^2 \le \max\{x^2, y^2\}$ for any $x < s < y$, we have
\begin{align}
\gamma(s) = \frac{\beta}{(s^2 + H^2)^{\frac{\alpha}{2}}} > \frac{\beta}{(\max\{x^2, y^2\} + H^2)^{\frac{\alpha}{2}}}.
\end{align}

As $p^*(s) > 0$, we have $\frac{1}{\gamma_0} > \frac{1}{\gamma(s)}$ holds for all $x < s < y$. To guarantee this, $\frac{1}{\gamma_0}$ must be larger than or equal to the maximum value of $\frac{1}{\gamma(s)}$, i.e.,
\begin{align}
\frac{1}{\gamma_0} \ge \frac{(\max\{x^2, y^2\} + H^2)^{\frac{\alpha}{2}}}{\beta}. \label{eq:gamma1}
\end{align}

On the other hand, according to (\ref{eq:powerconstr}), i.e.,
\begin{align}
\frac{1}{v} \int_{x}^{y} p^*(s) \mathrm d s &= \frac{1}{v} \int_{x}^{y} \bigg(\frac{1}{\gamma_0} - \frac{(s^2 + H^2)^{\frac{\alpha}{2}}}{\beta} \bigg) \mathrm d s \nonumber\\
&= \frac{y-x}{v} \frac{1}{\gamma_0} - \frac{1}{v} \int_{x}^{y} \frac{(s^2 + H^2)^{\frac{\alpha}{2}}}{\beta} \mathrm d s \nonumber\\
&\le E, \label{eq:vlb}
\end{align}
we have
\begin{align}
\frac{1}{\gamma_0} \le \frac{vE}{y-x} + \frac{1}{y-x} \int_{x}^{y} \frac{(s^2 + H^2)^{\frac{\alpha}{2}}}{\beta} \mathrm d s. \label{eq:gamma2}
\end{align}
According to (\ref{eq:gamma1}) and (\ref{eq:gamma2}), we have
\begin{align}
\frac{(\max\{x^2, y^2\} + H^2)^{\frac{\alpha}{2}}}{\beta} \le \frac{vE}{y-x} + \frac{1}{y-x} \int_{x}^{y} \frac{(s^2 + H^2)^{\frac{\alpha}{2}}}{\beta} \mathrm d s, \label{eq:cond2}
\end{align}
which is equivalent to (\ref{eq:cond}). Therefore, the sufficiency is proved.

In addition, to maximize the throughput, (\ref{eq:vlb}) must be satisfied with equality, which results in equality condition in (\ref{eq:gamma2}). Hence, (\ref{eq:gamma0}) is obtained.

\emph{2) Necessity:}

Suppose (\ref{eq:cond}) (or equivalently (\ref{eq:cond2})) holds true. We let $\frac{1}{\gamma_0}$ equals to the right hand side of (\ref{eq:gamma2}). Then, (\ref{eq:vlb}) is satisfied with equality, which guarantees that the power allocation is optimal as all the energy budget is fully utilized. Based on (\ref{eq:cond2}) and the equality of (\ref{eq:gamma2}), we have
\begin{align}
\frac{1}{\gamma_0} \ge \frac{(\max\{x^2, y^2\} + H^2)^{\frac{\alpha}{2}}}{\beta} > \frac{(s^2 + H^2)^{\frac{\alpha}{2}}}{\beta}
\end{align}
for all $x < s < y$. Therefore, we have $p^*(s) = \frac{1}{\gamma_0} - \frac{(s^2 + H^2)^{\frac{\alpha}{2}}}{\beta} > 0$ for $x < s < y$, and hence, the necessity is proved.

The optimal throughput can be obtained by substituting $p(s)$ in (\ref{eq:thrmaxobj}) with (\ref{eq:ps}) and deducing as follows
\begin{align}
B_{\max}(x,y,v) &= \frac{W}{2v}\int_x^y \log_2 \frac{\beta}{\gamma_0(s^2+H^2)^{\frac{\alpha}{2}}} \mathrm d s \nonumber\\
& = \frac{W}{2v} \bigg( s \log_2 \frac{\beta}{\gamma_0} \bigg|_x^y - \frac{\alpha}{2} \int_x^y \log_2 (s^2+H^2) \mathrm d s \bigg) \nonumber\\
& = \frac{W}{2v} \bigg( \Big(s \log_2 \frac{\beta}{\gamma_0} - \frac{\alpha}{2} s\log_2 (s^2+H^2) \Big)\bigg|_x^y + \frac{\alpha}{2} \int_x^y s \mathrm d \big( \log_2 (s^2+H^2) \big)\bigg) \nonumber\\
& = \frac{W}{2v} \bigg( s \log_2 \frac{\beta}{\gamma_0(s^2+H^2)^{\frac{\alpha}{2}}} \bigg|_x^y + \frac{\alpha}{\ln 2} \int_x^y \frac{s^2}{s^2+H^2} \mathrm d s \bigg)\nonumber\\
& = \frac{W}{2v} \bigg( s \log_2 \frac{\beta}{\gamma_0(s^2+H^2)^{\frac{\alpha}{2}}} \bigg|_x^y + \frac{\alpha}{\ln 2} \int_x^y \Big(1-\frac{H^2}{s^2+H^2} \Big) \mathrm d s \nonumber\\
& = \frac{W}{2v} \bigg(s \log_2\frac{\beta}{\gamma_0(s^2 + H^2)^{\frac{\alpha}{2}}} + \frac{\alpha s}{\ln2} - \frac{\alpha H}{\ln2} \arctan\frac{s}{H} \bigg) \bigg|_x^y. \label{eq:Bmax}
\end{align}

\section{Proof of Theorem \ref{prop:decv}} \label{proof:decv}
Based on (\ref{eq:thrmaxobj}) and (\ref{eq:powerconstr}), to achieve the maximum throughput, all the energy should be fully used, i.e.
\begin{align}
\int_{x}^{y} p(s) \mathrm d s = vE.
\end{align}
Replacing $p(s)$ in the above equation by (\ref{eq:ps}), we have
\begin{align}
\frac{1}{\gamma_0} = \frac{vE}{y-x} + \frac{1}{y-x} \int_{x}^{y} \frac{1}{\gamma(s)} \mathrm d s.
\end{align}
According to the first line of (\ref{eq:Bmax}), we have
\begin{align}
B_{\max}(x,y,v) &= \frac{W}{2v}\int_x^y \log_2 \frac{\gamma(s)}{\gamma_0} \mathrm d s \nonumber\\
&= \frac{W}{2v} \left( \int_x^y \log_2 \frac{1}{\gamma_0} \mathrm d s - \int_x^y \log_2 \frac{1}{\gamma(s)} \mathrm d s\right) \nonumber\\
& = \frac{W}{2v} \left( a_1 \log_2 \bigg(\frac{1}{a_1} \Big( {vE} + a_2\Big)\bigg) - a_3\right),
\end{align}
where
\begin{align}
a_1 &= y-x, \\
a_2 &= \int_{x}^{y} \frac{1}{\gamma(s)} \mathrm d s,\\
a_3 &= \int_x^y \log_2 \frac{1}{\gamma(s)} \mathrm d s.
\end{align}

Define a function
\begin{align}
g(u) = u \left( a_1 \log_2 \bigg(\frac{1}{a_1} \Big( \frac{E}{u} + a_2\Big)\bigg) - a_3\right).
\end{align}
Since
\begin{align}
g''(u) = -\frac{a_1E^2}{u(E+a_2u)^2} < 0
\end{align}
for all $u>0$, $g'(u)$ is a decreasing function of $u$. Therefore,
\begin{align}
g'(u) &= a_1 \log_2 \bigg(\frac{1}{a_1} \Big( \frac{E}{u} + a_2\Big)\bigg) - a_3 - \frac{a_1E}{E+a_2u}\nonumber\\
&> g'(+\infty) \nonumber\\
&= a_1 \log_2 \bigg(\frac{a_2}{a_1}\bigg) - a_3 \nonumber\\
&= (y-x) \log_2 \bigg(\frac{1}{y-x}\int_{x}^{y} \frac{1}{\gamma(s)} \mathrm d s\bigg) - \int_x^y \log_2 \frac{1}{\gamma(s)} \mathrm d s \nonumber\\
&\ge 0, \label{eq:gprime}
\end{align}
where the first inequality holds due to the monotonicity of $g'(u)$, and the second inequality holds due to the concavity of $\log$ function. Based on (\ref{eq:gprime}), we conclude that $g(u)$ is an increasing function of $u$. Since $B_{\max}(x,y,v) = \frac{1}{2}W g(\frac{1}{v})$, it is a decreasing function of $v$.

\section{Proof of Proposition \ref{prop:complex}} \label{proof:complex}
As $(x_n^*, y_n^*)$ is the optimal solution for the minimization problem in (\ref{eq:DPn}), we have
\begin{align}
J_n(s_n) &= \min_{s_n \le x_n \le y_n \le S_{N+1}} \{g_n(x_n, y_n) + J_{n+1}(y_n)\} \nonumber\\
&= g_n(x_n^*, y_n^*) + J_{n+1}(y_n^*).
\end{align}

For a given $s_n' \in [s_n, x_n^*]$, as $s_n' \ge s_n$, we have $[s_n', S_{N+1}] \subseteq [s_n, S_{N+1}]$. Therefore,
\begin{align}
J_n(s_n') &= \min_{s_n' \le x_n \le y_n \le S_{N+1}} \{g_n(x_n, y_n) + J_{n+1}(y_n)\} \nonumber\\
&\ge \min_{s_n \le x_n \le y_n \le S_{N+1}} \{g_n(x_n, y_n) + J_{n+1}(y_n)\} = J_n(s_n).  \label{eq:Tsub}
\end{align}

Secondly, as $s_n' \le x_n^*$, we have $s_n' \le x_n^* \le y_n^* \le S_{N+1}$. Hence,
\begin{align}
J_n(s_n') &= \min_{s_n' \le x_n \le y_n \le S_{N+1}} \{g_n(x_n, y_n) + J_{n+1}(y_n)\} \nonumber\\
&\le g_n(x_n^*, y_n^*) + J_{n+1}(y_n^*) = J_n(s_n).  \label{eq:Tsub2}
\end{align}

Combining (\ref{eq:Tsub}) and (\ref{eq:Tsub2}), we prove that $J_n(s_n') = J_n(s_n)$ for all $s_n' \in [s_n, x_n^*]$.

\bibliographystyle{IEEEtran}
\bibliography{ref}

\end{document}